\documentclass[letterpaper,onecolumn]{IEEEtran}

\usepackage[OT2,T1]{fontenc}
\usepackage{amsmath,amssymb,amsfonts,amsthm}
\usepackage{algpseudocode}
\usepackage{algorithm}
\usepackage{array}
\usepackage{textcomp}
\usepackage{url}
\usepackage{graphicx}
\usepackage{xcolor}
\usepackage{cite}

\theoremstyle{definition}
\newtheorem{definition}{Definition}[section]

\newtheorem{remark}[definition]{Remark}
\newtheorem*{remark*}{Remark}

\theoremstyle{plain}
\newtheorem{theorem}[definition]{Theorem}
\newtheorem{corollary}[definition]{Corollary}

\newtheorem{proposition}[definition]{Proposition}

\newcommand{\vF}{{ \mathbb F }}

\DeclareMathOperator{\ev}{ev}

\begin{document}

\title{A New Class of Linear Codes}

\author{Akash Bhople, Giacomo Cherubini, Giacomo Micheli, and Tefjol Pllaha%
\thanks{A. Bhople, G. Micheli, and T. Pllaha are with the University of South Florida, 4202 E Fowler Ave, 33620 Tampa, US. E-mail: akashbhople@usf.edu; gmicheli@usf.edu; tpllaha@usf.edu.}%
\thanks{G. Cherubini is with Istituto Nazionale di Alta Matematica ``Francesco Severi'', Research Unit Dipartimento di Matematica ``Guido Castelnuovo'', Sapienza University of Rome, Piazzale Aldo Moro 5, I-00185, Roma. E-mail: cherubini@altamatematica.it.}}

\maketitle

\begin{abstract}
Let $n$ be a prime power, $r$ be a prime with $r\mid n-1$, and $\varepsilon\in (0,1/2)$. Using the theory of multiplicative character sums and superelliptic curves, we construct new  codes over $\vF_r$ having length $n$, relative distance $(r-1)/r+O(n^{-\varepsilon})$ and rate $n^{-1/2-\varepsilon}$. When $r=2$, our binary codes have exponential size when compared to all previously known families of linear and non-linear codes with relative distance asymptotic to $1/2$, such as Delsarte--Goethals codes.
Moreover, concatenating with a Reed-Solomon code we get a family of codes of length
$n$ and rate $n^{-1/(2n+2)-2\varepsilon/(n+1)}+O(n^{-1/(n+1)})$ and relative distance $1/2+O(n^{-\varepsilon})$.
This shows that, for a fixed length, the rate of the concatenation suggested in \cite{KT2024} of a Reed-Solomon and a Reed-Muller code can be made an order of magnitude smaller than a concatenation of a Reed-Solomon with a large dimensional Shadow code, while still keeping the regime of relative distance $1/2$.
Finally, we show that the square of a Shadow code behaves like a random code and the Shadow code itself has a decoding algorithm, which suggest that such class of codes has the potential to be interesting for cryptographic applications.
\end{abstract}

\begin{IEEEkeywords}
Linear Codes, Polynomials over Finite Fields, Multiplicative Characters.
\end{IEEEkeywords}

\section{Introduction}

Let $r,p$ be distinct primes, $\ell$ be a positive integer, and $q=p^\ell$. In this paper we construct codes over prime fields $\mathbb F_r$ using an auxiliary field $\vF_q$ of different characteristic. 

A \emph{code $\mathcal C$} of \emph{length $n$} is a subset  of  $\vF_r^n$. The \emph{weight} of a codeword $c\in \mathcal C$ is defined as the number of non-zero entries of $c$ and denoted by $w(c)$.
The \emph{minimum distance} of $\mathcal C$ is defined to be $d=d(\mathcal C)=\min\{w(x-y): \; x,y\in \mathcal C, \, x\neq y\}.$
The \emph{rate} of a code $\mathcal C \subseteq \vF_r^n$ is defined as $\frac{\log_r |\mathcal C|}{n}$.
An $\vF_r$-\emph{linear code}, or simply a \emph{linear code} (when the base field $\vF_r$ is understood) is a code that is also an
$\vF_r$-subspace of $\vF_r^n$. When the base field is understood, an $[n,k,d]$ linear code $\mathcal C$ is a linear code having length $n$, dimension $k$ and minimum distance $d$.
We say that a code is \emph{binary} when $\mathcal C\subseteq\vF_2^n$. 
If $a_n$ is a sequence of real number, we say that a property for $a_n$ holds \emph{definitively in $n$} if it holds for all but finitely many natural numbers $n$ (e.g. $a_n>b_n$ definitively in $n$).

Binary codes having minimum distance roughly $n/2$ have been a fundamental object of study in the theory of codes, see for example \cite{barg2006spectral,delsarte1975alternating,jiang2004asymptotic,macwilliams1977theory,pang2023new}. As observed in~\cite{pang2023new}, because of the Plotkin bounds, the case $d\geq n/2$ only allows for extremely small codes, and was settled by Plotkin and Levenstein \cite[Chapter 2, Theorem~8]{macwilliams1977theory} assuming a conjecture concerning the
existence of sufficiently large number of Hadamard matrices.
Let $A(n,d)$ be the largest size of a binary code having minimum distance $d$ and length $n$.
For $d\leq n/2$, and $n-2d=o(\sqrt{n})$ McEliece \cite[Ch. 17, Th. 38]{macwilliams1977theory} proved the bound
\begin{equation}\label{eq:McElieceBound}
A(n,d) \lesssim n(n-2d+2),
\end{equation}
i.e. $A(n,d) \leq (1+o(1))n(n-2d+2)$.
In the regime $n-2d\sim n^{1/3}$, results by Sidelnikov imply that the bound is asymptotically sharp \cite{sidelnikovmutual}.
For $n-2d=\Omega(\sqrt n)$ the problem remains open
(given two positive sequences $f_n,g_n$, recall that $f_n=\Omega(g_n)$ if $f_n>Cg_n$ definitively in $n$, for some constant $C>0$).

A family of non-linear codes that achieves the largest known rate was studied by Delsarte and Goethals in \cite{delsarte1975alternating} (see also \cite[Theorem 1]{hergert1990delsarte} for a more explicit description of the parameters of the codes).
Let $m\geq 4 $  and $s\leq m/2-1$ be positive integers.
Delsarte--Goethals codes are binary non-linear codes of length $2^m$, size $2^{s(m-1)+2m}$ and minimum distance $2^{m-1}-2^{m/2-1+s}$. 
The rate  of these codes is therefore
\begin{equation}\label{intro:DGrate}
\frac{s(m-1)+2m}{2^m},
\end{equation}
which decays exponentially in $m$.
The paper \cite{pang2023new} also presents some new constructions
having rate roughly of the same quality as in \eqref{intro:DGrate}.
Delsarte--Goethals codes have been widely studied in the literature since the '70s (see for example \cite{hazewinkel2011encyclopaedia,hergert1990delsarte,leducq2012new,pang2023new}). 

Generally, Delsarte--Goethals codes are considered to be very good for their minimum distance \cite{hammons1994z}, and therefore also found applications in the context of compressed sensing (see for example \cite{BargRestricted2015}).
In this paper, we show that Delsarte--Goethals codes are by far not optimal, and that we can construct much better codes in the same regime of parameters that are also linear.
Indeed, below we present an exponential improvement to the state of the art in the regime $d<n/2$, when $n-2d=\Omega(\sqrt{n})$: we construct linear binary codes of essentially exponential size in the size of 
Delsarte--Goethals codes (Section~\ref{sec:construction}).
Moreover, in the non-binary case, we produce linear codes with rate $1/\sqrt n$
and distance $(r-1)n/r$, where $r$ is the size of the field where our codes are defined.
As a byproduct of our construction, we show an application of the theory of error correcting codes to number theory (Section \ref{sec:thmcharsum}).

The strategy of our construction is to take a small multiplicative subgroup
of prime order $r$ of a large field $\vF_q$ and use it to construct a code  over $\vF_r$ of length~$q$.
In order to achieve this, we start with a collection of superelliptic curves $\mathcal{S}:y^r=f(x)$,
where~$f$ varies within a certain (multiplicative) family indexed by $\vF_r$. Then, the entries of a codeword
are determined from the data of the ``shadow'' of $f$
through an $r$-multiplicative character~$\chi_r$, i.e.~the entry indexed by $x\in \vF_q$ is $\chi_r(f(x))$.
For example, when $r=2$ the multiplicative character loses all the information about $f(x)$
except for its quadratic class: the projection of the $(x,y)$ points of $\mathcal S$
on the $x$-axis (which we like to think as the shadow of $\mathcal S$) consists of the zeroes of the codeword attached to $f$
and the rest of the entries will be one. Finally, by simply changing notation from multiplicative to additive we obtain the desired codes.

Implementing rigorously the above strategy leads to a constructive proof of the following result.

\begin{theorem}\label{thm:main}
Let $q\geq 3$ be a prime power and $r$ be a prime dividing $q-1$. Then, for any $\varepsilon\in(0,1/2)$ and for any $q$ larger than an 
explicit constant $C=C(\varepsilon)$ depending only on $\varepsilon$, there is an  $\vF_r$-linear code having parameters
$[q,\lfloor q^{1/2-\varepsilon}\rfloor,d]$ with $d$ at least $(r-1)q/r-2q^{1-\varepsilon}$.
\end{theorem}

By picking $r=2$, an immediate consequence of Theorem \ref{thm:main}
is the existence of binary codes with certain parameters.

\begin{corollary}\label{cor:important}
Let $\varepsilon\in (0,1/2)$. There is an explicit constant $C(\varepsilon)$ such that for any prime power $q\geq C(\varepsilon)$,
there exists an  $\vF_2$-linear code having parameters
$[q,\lfloor q^{1/2-\varepsilon}\rfloor,d]$ with $d$ at least $q/2-2q^{1-\varepsilon}$.
\end{corollary}

The interesting regime of our codes is when the base field $\vF_r$ is small compared to the length of the code,
with the extremal case $r=2$ being particularly relevant:
from Corollary \ref{cor:important} we see that our codes are an ``exponential'' improvement of the Delsarte--Goethals codes in the asymptotic regime $d\sim n/2$ and $n-2d=\Omega(\sqrt{n})$. In other words, if one 
requires that the codes have asymptotic distance $n/2$
but also allows $n-2d$ to be large,
then our codes have size that is essentially exponential than Delsarte--Goethals code size (see for example Proposition \ref{prop:comparison}).
In fact, we can even restrict to the  exact same length $2^m$
and still have exponential size, as the following corollary shows.

\begin{corollary}\label{cor:compdelsgoeth}
Let $\varepsilon\in (0,1/2)$. Let $m$ be a positive integers and $q$ be the smallest odd prime power greater than $2^m$.
For every $0<\delta<5/12$, and definitively in $m$, there is an  $\vF_2$-linear code having parameters
$[2^m,\lfloor q^{1/2-\varepsilon }\rfloor,d]$, with $d\geq q/2-q^{1-\varepsilon}-2^{(\frac{7}{12}+\delta)m}$.
\end{corollary}

Using the above corollary we can beat Delsarte--Goethals codes for infinitely many parameters (not only asymptotically), as we now explain. For simplicity, choose a constant $\nu\in (1/12,1/2)$ so that  $s=\nu m$ is an integer. There exists a Delsarte--Goethals code with length $2^m$, size $2^{\nu m^2-\nu m+2m}$, and minimum distance
$2^{m-1}-2^{(1/2+\nu) m-1}$. 
On the other hand, the code in Corollary \ref{cor:compdelsgoeth} has the same length, and for any $\varepsilon\in (0,1/2)$, minimum distance larger than  
$q/2-q^{1-\varepsilon}-2^{(\frac{7}{12}+\delta)m}>2^m/2-2^{(1-\varepsilon)m}-2^{(\frac{7}{12}+\delta)m}$. Choosing now $\delta<\nu-1/12$ and $\varepsilon\in(1/2-\nu,5/12-\delta)$ , so that $7/12+\delta<1-\varepsilon<1/2+\nu$, we get that, definitively in $m$, our code has sightly improved minimum distance compared with Delsarte--Goethals but size larger than $2^{\lfloor q^{1/2-\varepsilon }\rfloor}>2^{\lfloor 2^{m(1/2-\varepsilon) }\rfloor}$,
which is  much larger than $2^{\nu m^2-\nu m+2m}$,
which is the size of a Delsarte--Goethals code. Moreover, our codes are linear, have a decoding algorithm and therefore can be used in practice.
We will show in Subsection \ref{subsec:lowerbounds} that actually we are an exponential improvement over Delsarte--Goethals construction for all  parameters for which it makes sense to compare the codes (see Proposition ~\ref{prop:comparison}).

Our codes are the largest known codes with asymptotic relative distance $1/2$: for example, our codes are exponentially larger than Hadamard codes, for the same asymptotic relative distance $1/2$.

Finally, we consider the square (with respect to the Schur product) of our codes. In Theorem~\ref{T-square}, we show that its dimension grows quadratically, mimicking the behavior of a random code.
Codes with compact algebraic description and easy encoding with a high dimensional square---such as ours---are good candidates for cryptographic applications because they provide small signature size while immediately annihilating the possibility of some known attacks~\cite{A_Cauvreuretal,7942048}.
In contrast, the concatenation of Reed-Solomon and Reed-Muller codes fails to have a large square; see Theorem~\ref{T-con}.

To show effectiveness of our results, and to have more readable proofs, in many of our statements
the code parameters depend
on $q^{b-\varepsilon}$: these can always be replaced by $o(q^b)$.

\section{Shadow Codes Construction}\label{sec:construction}

Let $n,L$ be positive integers, $q\geq n$ be a prime power and $\{p_1,\dots,p_L\}=P \subset \mathbb{F}_q[x]$ be a set of monic irreducible polynomials of degree at least $2$. For a positive integer $r$, let $\mathbb U_r$
denote the group of $r$-th roots of unity in $\mathbb C$.
Let now $r$ be a prime dividing $q-1$ and
$\chi_r: \vF_q^* \longrightarrow \mathbb U_r$ be a multiplicative character of order $r$.
Also, let $x=(x_1,\dots,x_n)$ be a given point in $\mathbb{F}_q^n$ with all distinct components
and $\phi_r: \mathbb U_r\longrightarrow \vF_r$ be any fixed isomorphism between $\mathbb U_r$ and 
the additive group $(\vF_r,+)$, which can obviously be endowed with a field structure as well.

Consider the map
\begin{align*}
\ev^{x,P}_{\chi_r,n}:\mathbb{F}_r^L &\longrightarrow \mathbb{F}_r^n\\
(v_1,\dots,v_L) & \longrightarrow \left( \phi_r\left(\chi_r\Bigl(\prod^L_{i=1} p_i(x_j)^{v_i}\Bigr)\right)\right)_{j\in \{1,\dots,n\}}.
\end{align*}

Since $\chi_r$ has order $r$, the map is well defined:  if $\overline v_i=v_i+s_ir$ for some $s_i$, then 
\[
\chi_r\Bigl(\prod^L_{i=1} p_i(x_j)^{v_i+rs_i}\Bigr)
=
\chi_r\Bigl(\prod^L_{i=1} p_i(x_j)^{v_i}\Bigr)
\Bigl(
\chi_r\Bigl(\prod^L_{i=1} p_i(x_j)^{s_i}\Bigr)
\Bigr)^r
=
\chi_r\Bigl(\prod^L_{i=1} p_i(x_j)^{v_i}\Bigr).
\]

\begin{proposition}\label{prop:linear}
The map $\ev^{x,P}_{\chi_r,n}$ is $\vF_r$-linear.
\end{proposition}
\begin{proof}
Since $(\vF_r,+)$ is cyclic, it is enough to show addivity, i.e.~to show
that $\ev^{x,P}_{\chi_r,n}(y+z)=\ev^{x,P}_{\chi_r,n}(y)+\ev^{x,P}_{\chi_r,n}(z)$ for every $y,z\in\vF_r^{L}$.
Plainly, we have 
\begin{align*}\label{eq:linearity}
\ev^{x,P}_{\chi_r,n}(y+z)=& \left( \phi_r\left(\chi_r\Bigl(\prod^L_{i=1} p_i(x_j)^{y_i+z_i}\Bigr)\right)\right)_{j\in \{1,\dots,n\}}\\
=& \left( \phi_r\left(\chi_r\Bigl(\prod^L_{i=1} p_i(x_j)^{y_i}\Bigr)\chi_r\Bigl(\prod^L_{i=1} p_i(x_j)^{z_i}\Bigr)\right)\right)_{j\in \{1,\dots,n\}}\\
=& \left( \phi_r\left(\chi_r\Bigl(\prod^L_{i=1} p_i(x_j)^{y_i}\Bigr)\right)+\phi_r\left(\chi_r\Bigl(\prod^L_{i=1} p_i(x_j)^{z_i}\Bigr)\right)\right)_{j\in \{1,\dots,n\}}\\
=& \ev^{x,P}_{\chi_r,n}(y)+\ev^{x,P}_{\chi_r,n}(z)\qedhere
\end{align*}
\end{proof}

\begin{definition}\label{def:code}
Let $r$ be a prime, $q$ be a prime power such that $r\mid q-1$. Let $P$ be a finite subset of $\vF_q[x]$
of monic irreducible polynomials of degree at least $2$.
Let $x=(x_1,\dots,x_n)\in \vF_q^n$ be a point with all distinct components.
 The $(r,P)$-shadow code $\mathcal{C}(\chi_r,P)$ is the $\vF_r$-linear code $\ev^{x,P}_{\chi_r,q}(\mathbb{F}_r^L)$.
\end{definition}
In the definition above, we restrict to monic polynomials just for simplicity of notation.

From now on we will restrict to the length $n=q$,
although, to avoid confusion with the auxiliary field $\vF_q$,
in some occasions we will still denote the length by $n$ instead of $q$.

Consider the minimum distance of $\mathcal{C}(\chi_r,P)$. By construction of the code,
it equals the minimum number of entries different from $1$ in the vector
$(\dots,\chi_r(f(x_j)),\dots)_{j\in\{1,\dots,n\}}$, where $f$ varies in the products of the $p_i$'s with exponents in $\vF_r$,
and it is here that superelliptic curves come into play.

\begin{theorem}\label{thm:mindist}
Let $\mathcal{C}(\chi_r,P)$ be an $(r,P)$-shadow code as in Definition \ref{def:code}, with $n=q$. Let $L=|P|$.
The minimum distance of $\mathcal{C}(\chi_r,P)$ is at least \[\frac{r-1}{r}q-L\max\{\deg(p_i): i\in \{1,\dots, L\}\}\sqrt{q}.\]
\end{theorem}
\begin{proof}
We start by finding an explicit upper bound for the $\vF_q$-rational affine points of the superelliptic curve $\mathcal S$ defined by $y^r-f(x)=0$,
with $f=\prod^L_{i=1} p_i^{v_i}$, $v_i\in\{0,\dots,r-1\}$.
By the standard theory, the curve $\mathcal S$ is irreducible because $Y^r-f(x)$ is monic and irreducible in $\vF_q(x)[Y]$.
By the Hasse-Weil Theorem, the function field $\vF_q(x,y)$ of this curve has at most 
$q+1+2g\sqrt{q}$ places of degree one, where $g$ denotes the genus.
The affine  $\vF_q$-rational points of $\mathcal{S}$
are all non-singular and therefore can be bounded from above by the places of the function field $\vF_q(x,y)$.
Next, we estimate the genus $g$.
Using \cite[Proposition 6.3.1]{bib:stichtenoth2009algebraic} we get
\[g\leq \frac{r-1}{2} \Bigl(-1+\sum^{L}_{i=1} \deg(p_i)\Bigr) -\frac{\gcd(r,\deg(f))-1}{2}.\]
Discarding the negative terms and bounding each summand by their maximum, this leads to the cleaner estimate
\[g\leq \frac{rL}{2} \max\{\deg(p_i): i\in \{1,\dots, L\}\}.\]
Now we are ready to conclude the proof. First, observe that for every point $(x,y)$ of $y^r=f(x)$ there is a set of $r$ points $(x,\xi y)$, where $\xi $ is an $r$-th root of unity in $\vF_q$ (notice that $f(x)\neq 0$ by construction of the $p_i$'s).
Therefore, the number of $x$'s such that $\prod^{L}_{i=1}p_i(x)^{v_i}$ is an $r$-th power
equals the number of $\vF_q$-points of the curve $y^r=f(x)$ divided by $r$ and can be bounded by
$B=\frac{q}{r}+L\max\{\deg(p_i): i\in \{1,\dots, L\}\}\sqrt{q}$.
In turn, this gives an upper bound for the number of zeroes of a codeword in $\mathcal{C}(\chi_r,P)$
and so $q-B$ is a lower bound for the distance of the code.
\end{proof}

Using the lower part of the Hasse-Weil bound, analogous arguments
(but more technical, if one wants to keep track of constants)
show that the distance of the code is at most
$\frac{r-1}{r}q+O(L\sqrt{q})$.

Now that we know the length and the distance of $\mathcal{C}(\chi_r,P)$, we look at its dimension.

\begin{proposition}\label{prop:dimnondeg}
Let $\mathcal{C}(\chi_r,P)$ be an $(r,P)$-shadow code as in Definition \ref{def:code}, with $n=q$.
Let $L=|P|$ and $p_1,\dots,p_L$ be the elements in $P$.
If $\deg{p_i}< \frac{q(r-1)-1}{rL\sqrt{q}}$ for all $i\in \{1,\dots, L\}$, then the dimension of $\mathcal{C}(\chi_r,P)$ is $L$.
\end{proposition}
\begin{proof}
It is enough to show that the map $\ev^{x,P}_{\chi_r,q}$ is injective. If $\ev^{x,P}_{\chi_r,q}(v)=0$, then this means that the superelliptic curve
$y^r=\prod^L_{i=1} p_i(x)^{v_i}$ has at least $rq$ points, but the Hasse-Weil bound
combined with the assumption on $\deg{p_i}$ implies
\[rq\leq q+1+Lr\max_{i\in \{1,\dots, L\}}\{\deg{p_i}\}\sqrt{q}<rq,\]
which gives a contradiction.
\end{proof}

It is immediate to see that for $q\geq 3$ it is always possible to select degree $2$ (irreducible) polynomials, provided $L$ is not too large. Indeed, we obtain the following corollary, which is a direct consequence of Proposition \ref{prop:dimnondeg}.

\begin{corollary}
Let $P=\{p_1,\dots, p_L\}$ be a set of monic irreducible polynomials of degree $2$.
If $L< \frac{q(r-1)-1}{2r\sqrt q}$, then the dimension of $\mathcal{C}(\chi_r,P)$ is $L$.
\end{corollary}

It is worth noticing that we can essentially always select only degree two polynomials, since the regime of $L$ we are interested in is $o(\sqrt q)$ and the number of irreducible polynomials of degree two is $(q^2-q)/2$.

\begin{proof}[Proof of Theorem \ref{thm:main}]
We select $L=\lfloor q^{1/2-\varepsilon}\rfloor$ among the $(q^2-q)/2$
irreducible polynomials of degree $2$, say $p_1,\dots,p_L$. Then, we compute
$\ev^{x,P}_{\chi_r,q}(\vF_r^L)$
using such polynomials, getting a code of length $q$.
Observe that for $q$ greater than a constant depending only on $\varepsilon$,
the hypothesis of Proposition \ref{prop:dimnondeg} are verified:
\begin{equation}\label{eq:forL}
L=\lfloor q^{1/2-\varepsilon}\rfloor< \frac{\sqrt{q}}{4} -\frac{1}{4\sqrt{q}}\leq \frac{1}{2}\frac{q(r-1)-1}{r\sqrt{q}}.
\end{equation}
Hence, the dimension is $L$. Finally, the bound on the distance follows from Theorem \ref{thm:mindist}.
\end{proof}

\begin{remark}
The constant $C(\varepsilon)$ in Theorem \ref{thm:main} can be immediately extracted from the inequality \eqref{eq:forL}: select for example 
$q>C(\varepsilon)=6^{1/\varepsilon}$.
\end{remark}

\begin{proof}[Proof of Corollary \ref{cor:compdelsgoeth}]
Choose $q$ to be the next prime greater than $2^m$, which has size at most $2^m+2^{(7/12+\delta)m}$ by Huxley's work~\cite{huxley1972difference},
and construct the $[q,\lfloor q^{1/2-\varepsilon}\rfloor,q/2-q^{1-\varepsilon}]$ linear code guaranteed by Corollary \ref{cor:important}. Then, puncture the code at $t\leq 2^{(\frac{7}{12}+\delta)m}$ components, obtaining a new code $\mathcal C$ with distance at least $q/2-q^{1-\varepsilon}-2^{(\frac{7}{12}+\delta)m}$ but with length $2^m$ (and same dimension, because the distance asymptotic to $q/2$ guarantees non-degeneracy when one punctures at $2^{(\frac{7}{12}+\delta)m}\sim q^{(\frac{7}{12}+\delta)}=o(q)$ components).
\end{proof}

\section{Comparison with existing bounds}
Denote by $A_r(n,d)$ the maximal size of a $r$-ary code of length $n$ and distance~$d$.
Let us review upper and lower bounds on $A_r(n,d)$ and test them on
the code $\mathcal{C}(\chi_r,P)$ constructed in the previous section.

\subsection{Upper bounds}
The Hamming bound states that
\[
A_r(n,d) \leq \frac{r^n}{\sum_{j=0}^{t}\binom{n}{j}(r-1)^j},\qquad t=\left\lfloor\frac{d-1}{2}\right\rfloor.
\]
If $d\sim \frac{r-1}{r}n$ then $t\sim\frac{r-1}{2r}n$ and
the sum of binomial coefficients is asymptotic to its last term, which
can be estimated by $\geq r^{n\big(H_r(\frac{r-1}{2r})-o(1)\big)}$,
where $H_r$ is the $r$-ary entropy function, i.e.
\[H_r(x)=x\log_r(r-1)-x\log_rx-(1-x)\log_r(1-x).\]
This shows that the information rate satisfies
(cf.~also \cite[Proposition 8.4.3]{bib:stichtenoth2009algebraic})
\begin{equation}\label{hammingbound}
\frac{\log_rA_r(n,d)}{n} \leq 1-H_r\Bigl(\frac{r-1}{2r}\Bigr)-o(1).
\end{equation}
Let $\epsilon\in(0,1/2)$, $q$ be sufficiently large
and $L=\lfloor q^{1/2-\epsilon}\rfloor$, so that
the code $\mathcal{C}(\chi_r,P)$ has dimension $L$
by Proposition~\ref{prop:dimnondeg}. Hence, its rate is
\begin{equation}\label{S3:rate}
\frac{L}{q}\sim q^{-1/2-\epsilon}\sim n^{-1/2-\epsilon}
\end{equation}
(recall that $n=q$).
Comparing to \eqref{hammingbound} and observing that $H_r(\frac{r-1}{2r})\leq H_2(\frac{1}{4})\approx 0.8112$,
we see that \eqref{S3:rate} is off by a factor $n^{1/2+\epsilon}$.
At the same time,
we know that when the relative distance is bounded away from zero,
the Hamming bound is never sharp (see e.g.~\cite[Chapter 4]{venkatesan}).
In 2006, Barg and Nogin \cite{barg2006spectral} showed by spectral methods that one can
refine the Hamming bound in many situations. As an example, for binary codes their results
show \cite[p.~82]{barg2006spectral} that
\begin{equation}\label{eq:specboundgeneral}
A_2(n,d) \leq \frac{4(n-k)}{n-\lambda_k} \binom{n}{k}
\end{equation}
for all $k\leq n$ such that $\lambda_{k-1}\geq n-2d$,
where $\lambda_{k-1}$ (resp.~$\lambda_k$) is the maximal eigenvalue of the $k\times k$ self-adjoint matrix $S_{k-1}=(s_{ij})_{i,j=1}^{k}$
defined by $s_{i,i+1}=s_{i+1,i}=\sqrt{i(n+1-i)}$ for $1\leq i\leq k$ and $s_{ij}=0$ otherwise
(resp.~of the $(k+1)\times(k+1)$ matrix with $1\leq i\leq k+1$).
The number $\lambda_k$ is always positive and strictly less than $n$.
Estimates on $\lambda_k$ for $k$ fixed and $n\to\infty$ lead to the polynomial bound

\begin{equation}\label{spectralbound}
A_2(n,d)=O_k(n^k),
\end{equation}
provided $d\geq \frac{n}{2}-c_k\sqrt{n}$ for a suitable constant $c_k$,
see e.g.~\cite[(12)--(14)]{pang2023new} for more details.
Note that \eqref{spectralbound} is stronger than the Hamming bound
in this regime.

For what concerns us we observe that, if
$L\leq \lfloor n^{1/2-\epsilon}\rfloor$
and $k=(3L)^2$, then by \cite[Lemma~1 and (13)]{pang2023new}
and Theorem~\ref{thm:mindist} we have
$\lambda_k\leq (2+o(1))\sqrt{kn} = (6+o(1))L\sqrt{n}=O(n^{1-\epsilon})$ and
$\lambda_{k-1}\geq (1+o(1))\sqrt{kn}=(3+o(1))L\sqrt{n}\geq n-2d$.
Hence \eqref{eq:specboundgeneral} applies and yields (recall that $k=o(\sqrt{n})$)
\[
A_2(n,d)\leq (4+o(1)) \binom{n}{k}
\sim{} \left(\frac{ne}{k}\right)^k \frac{1}{\sqrt{2\pi k}} \exp\left(-\frac{k^2}{2n}\right)
\ll \exp\big(2k\log\frac{n}{k}\big),
\]
that is, the rate is at most $O(\frac{L^2}{n}\log n)$, as opposed to the
bound $O(1)$ obtained in \eqref{hammingbound}. In particular, recalling
that $\mathcal{C}(\chi_r,P)$ has rate $\frac{L}{n}$,
we see that we are off by a factor not larger than $O(L\log n)$,
meaning that the discrepancy between our construction
and the theoretical upper bound reduces as $L$ decreases.
In particular, we have the following.

\begin{corollary}
Let $A_2(n,d)$ be the maximal size of a binary code of length $n$ and distance $d<n/2$
with $|n-2d|\leq 2\sqrt{n}\log n$. Let $\chi$ be the quadratic character in $\mathbb{F}_q$
and let $\mathcal{C}(\chi,P)$ be the binary code constructed in
Definition \ref{def:code}, with $|P|=\log n$ and $\deg p=2$ for all $p\in P$. Then
\[\frac{\log A_2(n,d)}{\log |C(\chi,P)|}\ll (\log n)^2.\]
\end{corollary}

\subsection{Lower bounds: Gilbert--Varshamov bound, Delsarte--Goethals codes}\label{subsec:lowerbounds}
One of the most well-known lower bounds on $A_r(n,d)$ is the Gilbert--Varshamov bound
(see~\cite[Chapter 2]{venkatesan}), which is proved by probabilistic methods
and is a good testing device to measure the distance of a deterministic code
from a random code. The bound states that
\begin{equation}\label{S3:GV}
A_r(n,d) \geq \frac{r^n}{\sum\limits_{i=0}^{d-1}\binom{n}{i}(r-1)^i}.
\end{equation}
Denote by $A^*$ the quantity on the right-hand side above.
We are interested in the regime $d-1\geq \frac{r-1}{r}n-n^{1-\epsilon}$.
Let us show that in this regime we have $\frac{\log A^*}{n}\ll n^{-1+\epsilon}$,
whereas our code has rate at least $n^{-1/2-\epsilon}$. This will imply
that our code has rate better than the GV bound by at least a factor $n^{1/2-2\epsilon}$.

The sum of binomial coefficients in \eqref{S3:GV}
can be estimated crudely by keeping only the last term.
We write $d-1=np$ with $p=\frac{r-1}{r}-\frac{1}{n^\epsilon}$ and obtain
\[
\sum_{i=0}^{d-1}\binom{n}{i}(r-1)^i \geq \binom{n}{np}(r-1)^{np} \geq r^{nH_r(p)}\frac{1}{\sqrt{8p(1-p)}}.
\]
This shows that
\[A^*\leq r^{n(1-H_r(p))}\sqrt{8p(1-p)}\leq \sqrt{2}r^{n^\epsilon}.\]
On taking logarithms and dividing by $n$,
we deduce that $\frac{\log A^*}{n}\ll n^{-1+\epsilon}$, as claimed.
On the other hand, by picking $L=n^{1/2-\epsilon}$ and selecting polynomials $p_1,\dots,p_L$
of degree 2, the code $\mathcal C(\chi_r,P)$ from Section \ref{sec:construction} has dimension $L$
and therefore rate $n^{-1/2-\epsilon}$, proving our second claim.

Let $s_m$ be a sequence of natural numbers such that $s_m\leq m/2-1$. 
The choice of parameter $s=s_m$ in the Delsarte--Goethals construction produces  a sequence of codes
in the regime $d/n\sim 1/2$ and $n-2d=\Omega(\sqrt{n})$ 
(recall that $n=2^m$ and $d=2^{m-1}-2^{m/2-1+s_m}$).
Notice that the situation $n-2d=o(\sqrt{n})$ is already dealt with using other constructions such as Sidelnikov's, which achieves McEliece optimality in \eqref{eq:McElieceBound}.
We emphasize that the McEliece bound only allows for very small codes, i.e.~at most quadratic in the length.
 
The following proposition shows a wide set of regimes in which we provide exponentially larger codes than Delsarte--Goethals codes.

\begin{proposition}\label{prop:comparison}
Let $\delta\in(0,5/12)$ be a real number.
Let $s_m$ be a sequence of natural numbers  such that  $\delta m+1\leq s_m\leq m/2-1$. 
Let $\{\mathcal D_m\}$ be a sequence of Delsarte--Goethals codes with
length $2^m$, size $2^{s_m(m-1)+2m}$ and distance $d(\mathcal{D}_m)=2^{m-1}-2^{m/2-1+s_m}$.
There exists a sequence of shadow codes $\{\mathcal C_m\}$, each of length a prime power $q_m\sim 2^m$, with $q_m>2^m$,  such that  
\[\log_2(|\mathcal C_m|)/\log_2(|\mathcal D_m|)=\Omega\Bigl(\frac{2^{m\delta/2}}{m^2}\Bigr),\]
and $\mathcal C_m$ verifies 
\begin{equation}\label{eq:conditionn2d}
d(\mathcal D_m)\leq d(\mathcal C_m), 
\end{equation}
definitively in $m$.
\end{proposition}

\begin{remark}
The theorem above shows that we can construct a sequence of codes with better distance, and that have exponential parameters in Delsarte--Goethals' codes parameters.
\end{remark}
\begin{proof}
We will take $\mathcal{C}_m$ to be the code provided by Corollary \ref{cor:important},
with $q=q_m$ and a suitable choice of $\varepsilon\in(0,1/2)$.
First, fix $q_m$ to be the smallest prime power greater than $2^m$,
which has size at most $2^{m} + 2^{m(7/12+\delta)}$ by \cite{huxley1972difference},
definitively in $m$.
To verify \eqref{eq:conditionn2d}, it is enough to show that
$q_m-4q^{1-\varepsilon}-2^m+ 2^{m/2+s_m}\geq 0$
definitively in $m$. Certainly, $q_m-2^m\geq 0$ by definition
of $q_m$, so we can focus on proving $2^{m/2-1+s_m}\geq 2q_m^{1-\varepsilon}$.
We need to show
\[m/2-1+s_m\geq (1-\varepsilon)\log_{2}(q_m)+1.\]
Since $\delta<5/12$, we have $q_m\leq 2^{m} + 2^{m(7/12+\delta)}<2^{m+1}$.
Therefore, definitively in $m$, we can estimate
\[m/2-1+s_m\geq m(1/2+\delta)\geq  (1-\varepsilon)(m+1)+1\geq (1-\varepsilon)\log_{2}(q_m)+1\]
provided $1/2+\delta>1-\varepsilon$, which forces $\varepsilon\in (1/2-\delta,1/2)$.
Setting $\nu=1-\delta$ and $\varepsilon=\nu/2$, we obtain
\[
\frac{\log_2(|\mathcal C_m|)}{\log_2(|\mathcal D_m|)}\sim
\frac{q_m^{1/2-\varepsilon}}{s_m(m-1)+2m}
\geq 
\frac{2^{m\delta/2}}{m(m-1)/2+2m},
\]
from which the conclusion follows.
\end{proof}

\section{An application of Coding Theory to Number Theory: the maximum value of certain character sums}\label{sec:thmcharsum}

Finally, we apply the Plotkin bounds and the McEliece bound
to obtain two results on exponential sums.
\begin{theorem}
Let $q,L$ and $P=\{p_1,\dots,p_L\}$ satisfy the conditions in Proposition~\ref{prop:dimnondeg}.
Let $\chi$ be the quadratic character on $\mathbb{F}_q$.
Moreover, let $3+\log_2 q\leq \ell\leq L$.
Then there exist $a_1,\dots,a_\ell\in\{0,1\}$, not all zero, such that, on setting $f=\prod_{i=1}^{\ell} p_i^{a_i}$,
we have
\[
\sum_{x\in\mathbb{F}_q} \chi(f(x)) > 1.
\]
Similarly, assuming $\ell\geq \frac{3}{2}\log_2q$, we have
\begin{equation}\label{0301:eq001}
\max_{f}\Bigl|\sum_{x\in\mathbb{F}_q} \chi(f(x))\Bigr| = \Omega(\sqrt{q}),
\end{equation}
where $f$ ranges over the set $T=\big\{\prod_{i=1}^{\ell}p_i^{e_i}:\;e_i\in\{0,1\}\;\forall\;i=1,\dots,\ell\big\}\setminus\{1\}$.
\end{theorem}
\begin{proof}
We can construct a subcode of our code by selecting a subset of $\ell$ polynomials and taking the span $\mathcal D$ of the image of such polynomials under the map $\ev^{x,P}_{\chi_{2,q}}$.
The code $\mathcal D$ is a code over $\vF_2^q$ and has minimum distance at least the distance of our code, which is 
\begin{equation}\label{0301:eq002}
d=\min_{f\in T} \sum_{x\in \vF_q}\frac{1-\chi(f(x))}{2}=\frac{q}{2} - \frac{1}{2}\max_{f\in T}\sum_{x\in\mathbb{F}_q}\chi(f(x)).
\end{equation}
By contradiction, assume
\[\max_{f\in T} \sum_{x\in \vF_q}\chi(f(x)) \leq 1.\]
Then $2d+1\geq q$. From the Plotkin bounds \cite[Chapter 2, \S2]{macwilliams1977theory}
one deduces that if $2d+1\geq q$ then $A_2(q,d)\leq 4d+4\leq 4q+4$ since $d\leq q$.
On the other hand,  $A_2(q,d)=2^{\ell}$ by Proposition~\ref{prop:dimnondeg} and
recalling our choice of $\ell$ we obtain
\[
8q \leq 2^{\ell} = A_2(q,d) \leq 4q+4,
\]
which is impossible for all $q\geq 2$.

Now let us prove \eqref{0301:eq001}. Assume by contradiction that
\[
\max_{f} \Bigl|\sum_{x\in\mathbb{F}_q} \chi(f(x))\Bigr| = o(\sqrt{q}).
\]
Then by \eqref{0301:eq002} we get $q-2d=o(\sqrt{q})$ and by \eqref{eq:McElieceBound}
we deduce we should have $A_2(q,d)\leq (1+o(1)) q(q-2d+2)=o(q^{3/2})$.
The assumption on $\ell$ yields
\[
q^{\frac{3}{2}} \leq 2^{\ell} = A_2(q,d) = o(q^{3/2}),
\]
which is impossible for large $q$.
\end{proof}

\section{Code Boosts}\label{sec:boosts}

\subsection{Degree one Boost}
As Kschischang and Tasbihi showed in \cite{KT2024}, it is possible to gain a constant factor on the rate by
selecting degree one polynomials in Theorem~\ref{thm:mindist}
and considering a shorter code. This nice modification doubles the rate of the shadow code at the price of having a slightly more complicated construction.  The modification does not improve the order of convergence of the rate decay, which is still $q^{-1/2}$ at best as $q$ grows. In what follows, we propose a couple of boosts that mitigate asymptotically the order of convergence of the rate decay,
reaching up to any arbitrarily small negative power of the length.
For the sake of simplicity of discussion we still restrict to
irreducible polynomial of degree at least two, even though adaptations to
degree one are possible.

\subsection{Binary Concatenation Boost}
In the same paper \cite{KT2024}, Kschischang and Tasbihi showed
that it is possible to concatenate a Reed-Solomon code (RS) with a first-order Reed-Muller inner code (RM) to obtain a code that has better guaranteed distance than a Shadow code and same asymptotic rate and relative distance. In this section we show
that concatenating an RS code of length $q^n$ with a Shadow code of length $q$ produces a dimension boost to Theorem \ref{thm:main},
reaching dimension as large as $N^{1-\frac{1}{2n+2}-\frac{2\varepsilon}{n+1}}$ (with $N$ denoting the code length and fixed arbitrarily small $\varepsilon$),  still having relative distance asymptotic to $1/2$. In fact, for a fixed length, concatenation of RS and Shadow code will always produce a better rate than a concatenation of RS and RM,
with distance $d=N/2+O(N^{1-\varepsilon})$.
This happens because first order RM codes do not allow large dimension compared with length, whereas Shadow codes do, increasing the parameter range of the codes that are obtainable by concatenation (as we will explain at the end of the section) at the price of increasing $N-2d$ (but leaving the main term of $d$ untouched).

Fix $\varepsilon\in(0,\frac{1}{2})$ and any positive integer $n$.
Let $q$ be a prime power and set $m=\lfloor q^{1/2-\varepsilon}\rfloor$ and $K=\lfloor q^{n-\varepsilon}\rfloor$.
Choose $q$ large enough so that we have $2^m\geq q^n$.
We consider 
a Reed-Solomon code $\mathcal C$ over $\vF_{2^m}$ with parameters $[q^n,K,q^n-K+1]$
and a binary Shadow code of length $q$, dimension $m$ and distance $q/2+O(q^{1-\varepsilon})$,
as given by Theorem~\ref{thm:main} and Corollary~\ref{cor:important}.

The concatenation works as follows:

\[\vF_{2^m}^{K} \xrightarrow{\text{RS Encoding}}\vF_{2^m}^{q^n}\xrightarrow{\text{componentwise encoding with ShC}} \vF_2^{q^{n+1}} \]

The resulting code has length $N=q^{n+1}$ and dimension $Km=q^{n+1/2-2\varepsilon}+O(q^{n})$.
The distance $d$ of such a code is at least the product of the distances of the two original ones,
since a non-zero vector in the code has at least $q^n-K+1$ blocks of $q$ components that are non zero (thanks to the Reed-Solomon encoding and the injectivity of the encoding map for Shadow codes) and within each block there are at least $q/2+O(q^{1-\varepsilon})$ entries that are different from zero. This amounts to a total of
\[
(q^n-K+1)(q/2-O(q^{1-\varepsilon})) = q^{n+1}/2 + O\left(q^{n+1-\varepsilon} + Kq\right)
\]
non-zero entries. Recalling that $K=\lfloor q^{n-\varepsilon}\rfloor$ leads to
\[d\geq q^{n+1}/2-O(q^{n+1-\varepsilon}),\]
where the implied constant depends only on $\varepsilon$. 
We summarise this in a theorem.
\begin{theorem}
Let $\varepsilon\in (0,1/2)$ and let $n$ be a positive integer and $q$ an odd prime power.
There exists a binary code of length $q^{n+1}$ with minimum distance at least $d\geq q^{n+1}/2-O(q^{n+1-\varepsilon})$, and dimension $q^{n+1/2-2\varepsilon}+O(q^{n})$.
\end{theorem}

In other words, as $N=q^{n+1}\to\infty$ the relative distance $d/N$
of the concatenated code is asymptotic to $1/2$ and the dimension is
\[q^{n+1-1/2-2\varepsilon}+O(q^{n})=N^{1-\frac{1}{2n+2}-\frac{2\varepsilon}{n+1}}+O(N^{1-\frac{1}{n+1}}).\]

Such parameters cannot be achieved by concatenating a $[M,K,M-K+1]$ Reed-Solomon code defined over $\vF_{2^a}$ with a first order Reed-Muller code
with parameters $[2^a,a,2^{a-1}]$, as we now explain.
To make a fair comparison, say one wants a code of length roughly $q^{n+1}$,
then one requires $M2^{a}\sim q^{n+1}$.
Now, one would need to at least match the dimension growth of the concatenation of the Reed-Solomon code and the Shadow code:
$Ka=\Omega(q^{n+1/2-2\varepsilon})$.
The largest $K$ can grow is equal to the length of the Reed-Solomon code, i.e. $K\sim M$. So this gives 
\[M2^{a}\sim q^{n+1} \quad \text{and} \quad Ma=\Omega(q^{n+1/2-2\varepsilon}),\]
from which it follows that
\[aq^{n+1}\sim aM2^{a}=\Omega(2^aq^{n+1/2-2\varepsilon}).\]
Then, we must have that $2^aq^{-1/2-2\varepsilon}$ is bounded by $O(a)$
and so $2^{a-(1/2+2\varepsilon)\log q}=O(a)$, which implies that at the very least $a\leq (1/2+3\varepsilon)\log q$ (because certainly $\log(a)\leq \varepsilon\log(q)$, or the code is too long, by looking at $M2^{a}\sim q^{n+1}$) and therefore $2^a\leq Cq^{1/2+3\varepsilon}$.
Given the asymptotic restrictions of $2^a$, the only way to get $M2^a\sim q^{n+1}$ is to force $M=\Omega(q^{n+1-1/2-3\varepsilon})$. But this is impossible:
there is no Reed Solomon code of length $q^{n+1-1/2-3\varepsilon}$ defined over a field of size at most $2^a\leq Cq^{1/2+3\varepsilon}$ for $n\geq 1$ and $q$ large enough.

\begin{remark}
The reason why one has this advantage by replacing RM with a Shadow code (SC) is that SCs have a larger rate than a first order Reed Muller, which guarantees that for a fixed length replacing RM with SC will lead to a better rate. More in general, however one concatenates using RS and RM for a fixed length, one should always be allowed to replace RM with SC and get a better rate: concatenation with SC produces an order of magnitude larger rate code at the price of increasing $n-2d$, but leaving $d\sim n/2$. This is particularly convenient in the case of these codes, as the bottleneck in their parameters lies in the dimension.
\end{remark}

\subsection{Binary Non-Linear Boost}

It is possible to enlarge the message space of Shadow codes by switching to the context of non-linear codes. We sketch here the construction.
Let $q$ be an odd prime power and let
\[\mathcal M=\{f\in \vF_q[x]: \, \text{$f$ is irreducible of degree $m$}\}\]
be the message space. 
Encoding works exactly as for standard Shadow codes:
each polynomial $f\in \mathcal M$ gets evaluated at all elements of $\vF_q$, and then the quadratic character is applied. In other words, the code is the image of the map
\[\psi_{\text{enc}}: \mathcal M\longrightarrow \vF_2^q\]
\[f\mapsto (\phi_2(\chi(f(x))))_{x\in \vF_q}\]
By choosing $m=q^{1/2-\varepsilon}$ and applying the usual estimate on \[(\chi(f(x)))_{x\in \vF_q}+(\chi(g(x)))_{x\in \vF_q}=(\chi(f(x)g(x)))_{x\in \vF_q}\] one gets that $d(\psi_{\text{enc}}(f),\psi_{\text{enc}}(g))$ is at least $q/2+O(q^{1-\varepsilon}$) if $f\neq g$.
Clearly, $|\mathcal M|= q^m/m+O(q^{m/2})$.
By applying logarithm on both sides we obtain
$\log(|\mathcal M|)\sim q^{1/2-\varepsilon}\log(q)$ which improves the rate of our code by a factor of $\log(q)$.
\section{On the Square of Shadow Codes}
Here we focus on $r = 2$, that is, binary Shadow codes in $\vF_2^q$.
Recall that the Schur product of two vectors $x,y\in \vF_2^q$ is defined as $x\star y = (x_1y_1, x_2y_2,\ldots,x_ny_n)$.
Then, the {\em square} of a code $C$ is defined as $C\star C = \langle\{x\star y\mid x,y\in C\}\rangle_{\vF_2}$.
Note that if $\{g_1,\cdots, g_L\}$ is a basis for $C$ then $\{g_i\star g_j \mid 1\leq i \leq j \leq L\}$ is a generating set for $C\star C$.

The main goal of this section is to demonstrate that the square of a Shadow code grows fast, mimicking the behavior of random codes, which in turn has applications in cryptography~\cite{7942048}. 
Clearly ${\rm dim}(C\star C)\leq {\rm min}\{q,L(L+1)/2\}$, and a random code will meet the bound with high probability~\cite{Zemor-square}.
\begin{theorem}\label{T-square}
Let $C$ be a binary Shadow code of dimension $L$ generated by degree-two polynomials $p_1,\dots,p_L$. If $q > 16L^22^L$ then $\dim(C\star C) \geq L(L+1)/2$. 
\end{theorem}
\begin{proof}
Let $\{g_1,\cdots, g_L\}$ be a basis for $C$ originating from $p_1,\dots,p_L$, of degree $2$.
Recall that codewords $g_i$ are indexed by $\vF_q$ so we will write $g_i = (g_{i,a})_{a\in \vF_q}$.
Suppose that $\dim(C\star C) < L(L+1)/2$. 
That means that there is a non trivial dependency, that is, $\sum_{i\leq j}\lambda_{i,j} g_i\star g_j = 0$ and not all $\lambda_{i,j}$'s are zero.
In particular, this implies
\begin{equation}\label{e-Q}
\sum_{i\leq j}\lambda_{i,j} g_{i,a}g_{j,a} = 0, \quad \text{for all } a\in \vF_q.
\end{equation}
Let us consider the quadratic form $\tilde{Q}:\vF_2^L\rightarrow \vF_2$ defined as $\tilde{Q}(x) = \sum_{i\leq j}\lambda_{i,j}x_ix_j$, which is nonzero by assumption.
Define also $v_a = (g_{i,a})_{1\leq i \leq L} \in \vF_2^L$.
Due to~\eqref{e-Q}, we have $\tilde{Q}(v_a) = 0$ for all $a\in \vF_q$.
Next, we will transfer everything to $\{1,-1\}$. 
Throughout, variables $x,y$ will be identified via $y = (-1)^x$.
Let us consider the associated {\em sign} quadratic form $Q:\{1,-1\}^L\rightarrow \{1,-1\}$ defined as $Q(y) = (-1)^{\tilde{Q}(x)}$.
Let $w_a = (\chi(p_i(a))_{1\leq i \leq L} \in \{1,-1\}^L$ corresponding to $v_a$. 
With this notation (and due to~\eqref{e-Q} again) we have that $Q(w_a) = 1$ for all $a\in \vF_q$.

The Fourier expansion of $Q$ is
\begin{equation}
Q(y) = \sum_{S\subset [L]}\widehat{Q}(S) \prod_{i\in S}y_i, \quad \widehat{Q}(S) = \frac{1}{2^L}\sum_{y \in \{1,-1\}^L}Q(y)\prod_{i\in S} y_i.
\end{equation}
If we evaluate on $w_a$, we have
\begin{align}
Q(w_a) & = \sum_{S\subset [L]}\widehat{Q}(S) \prod_{i\in S}w_{i,a} \\
& = \sum_{S\subset [L]}\widehat{Q}(S) \prod_{i\in S}\chi(p_i(a))\\
& = \sum_{S\subset [L]}\widehat{Q}(S)\chi(p_S(a)).
\end{align}
where $p_S(x) = \prod_{i\in S}p_i(x)$, and in the last equation we have used the multiplicativity of $\chi$.
Summing over all $a\in \vF_q$ and swapping summations, we obtain
\begin{equation}
q = \sum_{a\in \vF_q} Q(w_a) = \sum_{S\subset[L]}\widehat{Q}(S)\sum_{a\in \vF_q}\chi(p_S(a)).
\end{equation}
In the above equation, $S = \emptyset$ contributes $q\widehat{Q}(\emptyset)$, and we obtain
\begin{equation}
q(1 - \widehat{Q}(\emptyset)) = \sum_{S\neq \emptyset}\widehat{Q}(S)\sum_{a\in \vF_q}\chi(p_S(a)).
\end{equation}
Using the following bounds 
\begin{itemize}
\item $\Big|\sum_{a\in \vF_q}\chi(p_S(a))\Big| < 2L\sqrt{q}$
\item $|\widehat{Q}(\emptyset)|\leq 1/2$
\item $|\widehat{Q}(S)| = 2^{-m}$ exactly $2^{2m}$ times, where $2m \leq L$ is the rank of $Q$, and 0 otherwise
\end{itemize}
(the Weil bound implies the first and the last two are well known in the theory of quadratic forms~\cite{Lidl_Niederreiter_1996}) we obtain
\begin{equation}
\frac{q}{2} \leq 2^{L/2}\cdot 2L \cdot \sqrt{q},
\end{equation}
and this concludes the proof.
\end{proof}
\begin{corollary}\label{C-bound}
If $L = O(\log(q))$ then $\dim(C\star C) = L(L+1)/2$.
\end{corollary}
\begin{remark}
While Corollary~\ref{C-bound} is a theoretical guarantee {\em for all instances}, the expected value of $\dim(C\star C)$ of a Shadow code $C$ of dimension $L$ is still $L(L+1)/2$ as long as $L = O(\sqrt{q})$.
The intuition behind this relies on the fact that character values approach uniform distribution as $q$ grows and the codewords of $C$ behave like random vectors.
For cryptographic applications, where a high rate is desired for security purposes, one can construct Shadow codes of dimension slightly above the natural bound of $\sqrt{q}$.
In the figure below, we show the best Shadow code (in terms of the dimension) found for all primes $q\leq 1000$.
\begin{figure}[htb!]
    \centering
    \includegraphics[width=0.75\linewidth]{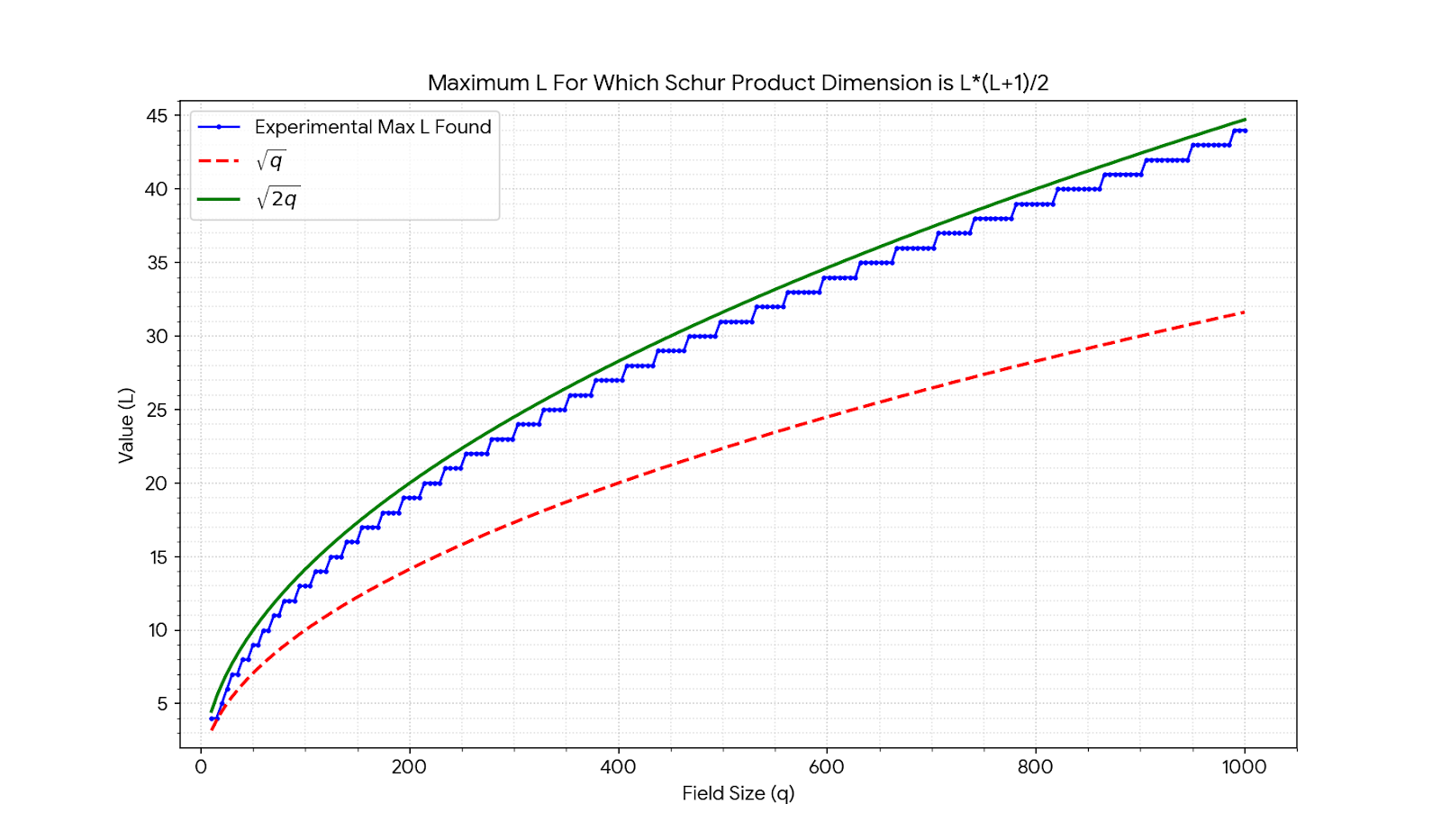}
    \caption{Best Shadow code found for primes $q\leq 1000$.}
    \label{fig:placeholder}
\end{figure}
As mentioned $\dim(C\star C)\leq q$, and so for the largest possible square code, we would want $L(L+1)/2 \sim q$, which asymptotically becomes $L\sim \sqrt{2q}$.
This experimental data also demonstrates that the square of Shadow code is close to the optimal bound $\sqrt{2q}$.
\end{remark}
\subsection{The Square Code of Concatenated RS-RM}
The goal of this subsection is to show that the square of a concatenated RS-RM code has low dimension, making them infeasible for cryptographic applications.

Let us first fix some notation.
Fix a vector space isomorphism $\Theta : \vF_2^{m+1}\rightarrow \vF_{2^{m+1}}$. 
The first order Reed-Muller code ${\rm RM}(1,m)$ is the image of the mapping 
\[
\begin{array}{ccccc}
\Phi_{\rm RM}&:&\vF_2^{m+1}&\longrightarrow &\vF_2^{2^m}\\
&&(v_0,\ldots, v_m)&\longmapsto&\left(v_0 + \sum_{j=1}^m v_jx_j\right)_{x\in \vF_2^m}.
\end{array}
\]
Define also a mapping
\[
\begin{array}{ccccc}
\Psi&:&\vF_{2^{m+1}}&\longrightarrow &{\rm RM}(1,m)\\
&&a&\longmapsto& \Phi_{\rm RM}(\Theta^{-1}(a)).
\end{array}
\]
For a Reed-Solomon code of length $N$ and dimension $K$, we take distinct elements $\beta_1,\ldots,\beta_N\in \vF_{2^{m+1}}$ and let
\[
C_{\rm RS} = \{(p(\beta_1),\ldots,p(\beta_N))\mid p(x)\in \vF_{2^{m+1}}[x], {\rm deg}(p(x)) < K \}.
\]
Then, the concatenated RS-RM code is defined as
\[
C_{\rm con} = \{(\Psi(s_1),\ldots,\Psi(s_N)) \in \vF_2^{N2^m}\mid (s_1,\ldots, s_N)\in C_{\rm RS}\},
\]
and has dimension $k = K(m+1)$.
\begin{theorem}\label{T-con}
Let $C_{\rm con}$ be the concatenated RS-RM code. Then $C_{\rm con}\star C_{\rm con} \subseteq \Big({\rm RM}(2,m)\Big)^N$. As a consequence
\[
{\rm dim}(C_{\rm con}\star C_{\rm con}) \leq N\cdot\left(1+m+{m\choose 2}\right) \sim N\frac{m^2}{2}.
\]
\end{theorem}
\begin{proof}
First, it is well-known that ${\rm RM}(1,m)\star {\rm RM}(1,m) = {\rm RM}(2,m)$, and the latter has dimension $1+m+{m\choose 2}$.
For two codewords $c_s =(\Psi(s_1),\ldots,\Psi(s_N)), c_t=(\Psi(t_1),\ldots,\Psi(t_N))\in C_{\rm con}$, since $\Psi(s_j)\star\Psi(t_j)\in {\rm RM}(1,m)\star {\rm RM}(1,m) = {\rm RM}(2,m)$ for all $j$, we have that $c_s\star c_t \in \Big({\rm RM}(2,m)\Big)^N$ and the first statement follows.
The last statement is immediate.
\end{proof}

\begin{remark}
If we fix $N = 2^m, K = \lfloor2^m/(m+1)\rfloor$, we have that the bound of Theorem~\ref{T-con} is vanishingly smaller than the universal bound $k(k+1)/2 \sim K^2m^2/2$.
\end{remark}

\section{DECODING BINARY SHADOW CODES}
Binary shadow codes are constructed by evaluating the quadratic residue character $\chi$ on a set of low-degree polynomials over a finite field $\mathbb{F}_q$. Consequently, decoding a binary shadow code from channel noise is mathematically equivalent to recovering an unknown, squarefree message polynomial $g(X)$ from noisy evaluations of the composite function $\chi(g(x))$.

Let $q$ be an odd prime power. Let $g(X) \in \mathbb{F}_q[X]$ be a monic squarefree polynomial. In this section, we present an efficient algorithm developed by Swastik Kopparty \cite{kopparty2026recovering} to recover $g(X)$ given access to a noisy version $v$ of $\chi \circ g$.

\begin{algorithm}[H]
\caption{Kopparty's Decoding Algorithm}
\label{algo:algorithm_a}
\begin{algorithmic}[1]
\Statex \textbf{Parameters:} degree $d\le O(\epsilon\sqrt{q})$ error-bound $e\le(\frac{1}{8}-\epsilon)q.$
\Statex \textbf{Input} $v:\mathbb{F}_{q}\rightarrow\{0,\pm1\}$

\State Set
    \begin{itemize}
        \item $M=\frac{16}{\epsilon}d$
        \item $c=\frac{M}{2}$
        \item $h=2\cdot e.$
        \item $D=d\cdot((q-1)/2+M)+cq=(1+O(\epsilon))\cdot\frac{1}{2}\cdot M\cdot q,$
        \item $u=h+dM=2e+O(\frac{d^{2}}{\epsilon}).$
    \end{itemize}
\State Solve an $\mathbb{F}_{q}$ system of linear equations to find polynomials $F(X),U_{0}(X),...,U_{M-1}(X)\in\mathbb{F}_{q}[X]$, not all zero, where:
    \begin{itemize}
        \item $\deg(F)\le D$
        \item for each $l$ $,\deg(U_{l})\le u$
        \item For all $\alpha\in\mathbb{F}_{q}$, $0\le l<M;$
        \[
        F^{[l]}(\alpha)=v(\alpha)\cdot U_{l}(\alpha).
        \]
    \end{itemize}
\State Factor $F(X)$ into irreducible factors:
    \[
    F(X)=\lambda\prod_{j}H_{j}^{\mu_{j}}(X)
    \]
    where the $H_{j}(X)\in\mathbb{F}_{q}[X]$ are distinct and monic, and $\lambda\in\mathbb{F}_{q}^{*}.$
\State Define
    \[
    J=\left\{j \text{ such that } \mu_{j}\in\left[\frac{3}{8}q,\frac{7}{8}q\right] \pmod q\right\}.
    \]
\State Set
    \[
    f(X)=\prod_{j\in J}H_{j}(X).
    \]
\State Return $f(X)$
\end{algorithmic}
\end{algorithm}

\subsection*{Acknowledgements}
This work has been supported by the National Science Foundation, grant numbers 2127742, 2338424. We would like to thank Daniele Bartoli, Alessandro Neri, and Paolo Santini, for sharing their useful comments on the preliminary versions of this manuscript.
We would like to thank Will Sawin who provided us with an alternative proof of the results in Section \ref{sec:thmcharsum}.
We would like to thank Lukas Koelsch for helpful discussions on Theorem~\ref{T-square}.

\nocite{*}
\bibliographystyle{IEEEtran}
\bibliography{biblio}

\end{document}